\documentclass[11pt]{article}
\pdfoutput=1
\usepackage{hyperref}

\usepackage{amsmath,amsthm,amssymb}

\usepackage[small,bf]{caption}

\usepackage{graphicx}

\hypersetup{colorlinks=true, linkcolor=black, citecolor=black,%
pdftitle={Infectious Random Walks},%
pdfauthor={Alberto Pettarin, Andrea Pietracaprina, Geppino Pucci, and Eli Upfal}}

\newcommand{\newterm}[1]{\emph{#1}}

\newcommand{\bigO}[1]{O\left(#1\right)}
\newcommand{\bigOt}[1]{\tilde{O}\left(#1\right)}
\newcommand{\bigTh}[1]{\Theta\left(#1\right)}
\newcommand{\bigTht}[1]{\tilde{\Theta}\left(#1\right)}
\newcommand{\bigOm}[1]{\Omega\left(#1\right)}
\newcommand{\bigOmt}[1]{\tilde{\Omega}\left(#1\right)}
\newcommand{\card}[1]{\left|#1\right|}
\newcommand{\abs}[1]{\left|#1\right|}
\newcommand{\dist}[1]{||#1||}
\newcommand{\prob}[1]{\mathrm{Pr}\left({#1}\right)}
\newcommand{\expe}[1]{\mathrm{E}\left({#1}\right)}
\newcommand{\vari}[1]{\mathrm{Var}\left({#1}\right)}

\newcommand{\sqrtn}{\sqrt{n}}

\newcommand{\isla}[3]{I_{#3}(#1; #2)}

\newcommand{\bs}{Broadcasting scenario}
\newcommand{\gs}{Gossiping scenario}
\newcommand{\bt}{T_{\mathrm{B}}}
\newcommand{\gt}{T_{\mathrm{G}}}

\newcommand{\B}{\mathcal{B}}
\newcommand{\I}{\mathcal{I}}
\newcommand{\Z}{\mathbb{Z}}
\newcommand{\dens}[1]{p^{#1}(n)}
\newcommand{\point}[1]{#1} 

\newcommand{\tf}[1]{\small{#1}}

\newtheorem{defi}{Definition}
\newtheorem{theo}{Theorem}
\newtheorem{lemm}{Lemma}

\begin{document}
\title{Infectious Random Walks%
	\thanks{Support for the first three authors was provided, in part, by MIUR of Italy
	under project AlgoDEEP, and by the University of
	Padova under the Strategic Project STPD08JA32 and Project
	CPDA099949/09.
	This work was done while the first author was visiting the Department of Computer Science of Brown University,
	partially supported by ``Fondazione Ing.~Aldo Gini'', Padova, Italy.}%
}

\author{
Alberto Pettarin \and Andrea Pietracaprina \and Geppino Pucci\\
	\tf{Department of Information Engineering,} 
	\tf{University of Padova}\\
	\tf{\texttt{\{pettarin,capri,geppo\}$@$dei.unipd.it}}
\and%
Eli Upfal\\
	\tf{Department of Computer Science,} 
	\tf{Brown University}\\
	\tf{\texttt{eli$@$cs.brown.edu}}
}

\date{}
\maketitle{}

\vspace{-0.5cm}
\begin{abstract}
We study the dynamics of information (or virus) dissemination
by $m$ mobile agents performing independent random walks
on an $n$-node grid.
We formulate our results in terms 
of two scenarios: \emph{broadcasting} and \emph{gossiping}.
In the broadcasting scenario,
the mobile agents are initially placed uniformly at random
among the grid nodes.
At time $0$, one agent is informed of a rumor
and starts a random walk.
When an informed agent meets an uninformed agent,
the latter becomes informed and starts a new random walk.
We study the \emph{broadcasting time} of the system, that is,
the time it takes for all agents to know the rumor.
In the gossiping scenario, each agent is given a distinct rumor
at time $0$ and all agents start random walks.
When two agents meet, they share all rumors they are aware of.
We study the \emph{gossiping time} of the system,
that is, the time it takes for all agents to know all rumors.
We prove that both the broadcasting and the gossiping times are
$\bigTht{n / \sqrt{m}}$ \emph{w.h.p.},
thus achieving a tight characterization up to logarithmic factors.
Previous results for the grid provided bounds which were weaker
and only concerned average times.
In the context of virus infection, a corollary of our results
is that static and dynamically moving agents
are infected at about the same speed.
\end{abstract}

\newpage

\section{Introduction}
The dynamics of multiple random walks moving in a common domain
provides an attractive combinatorial framework for studying diffusion processes
such as rumor spreading and virus infection.
In this work we consider two related scenarios on  an $n$-node square grid (or torus) $G_n$.

\vspace{0.1cm}
\noindent{\bf \bs:}
		$m$ agents are initially placed uniformly and independently at random in the nodes of $G_n$.
		At time 0 one agent is informed of a rumor and  starts a random walk.
		When an informed agent  meets an uninformed agent at a node,
		the latter becomes informed of the rumor and starts its own random walk.
		How long does it take until all agents are informed (\newterm{broadcasting time})?

\vspace{0.1cm}
\noindent{\bf \gs:}
		$m$ mobile agents start independent random walks
		from $m$ nodes of $G_n$ chosen uniformly and independently at random.
		At time 0 each agent has a distinct rumor.
		When two agents meet, each receives a copy of all  rumors carried by the other agent.
		How long does it take until all of the $m$ agents
		have received all of the $m$ rumors  (\newterm{gossiping time})?

\vspace{0.1cm}
\sloppy{
We prove that both the broadcasting and the gossiping times are $\bigTht{n/\sqrt{m}}$
with high probability,}
thus achieving a tight characterization
(up to polylogarithmic factors in $n$) of the complexity of rumor spreading
in both scenarios. In the context of virus infection, a corollary of our results
is that static agents placed at random locations and
dynamic agents moving as independent random walks
are infected at about the same speed.

Although at a first glance the two models look similar,
we do not have a simple reduction between them, hence we
develop separate upper and lower bound proofs for each scenario.
While the general approach employed in both cases is  similar,
there are some important differences in addressing the time dependencies,
since in the first scenario random walks hit static agents,
while in the second, random walks of informed agents need
to collide with the random walks of the uninformed ones.

It is not hard to verify a $\bigTht{{n^2}/{{m}}}$ tight bound
for the broadcasting and gossiping times on an $n$-node line or ring.
Combined with our results, this yields a $\bigTht{({n^2}/{{m}})^{1/d}}$ bound for $n$ node
$d$-dimensional grids, with $d=1, 2$. It remains open whether 
this relation generalizes to finite grids of higher dimensions.

\subsection{Related work}
Information dissemination has been extensively studied in the
literature under a variety of scenarios and 
objectives. Due to space limitations, we restrict our attention to
the results more directly related to our work.

A  prolific line of research has addressed broadcasting and gossiping in static graphs,
where the nodes of the graph represent active entities which exchange messages
along incident edges according to specific protocols
(e.g., \emph{push}, \emph{pull}, \emph{push-pull}).
The most recent results in this area relate the
performance of the protocols to expansion properties of the underlying
topology, with particular attention to the case of social networks,
where broadcasting is often referred to as \emph{rumor spreading}
\cite{ChierichettiLP10}.
(For a relatively recent, comprehensive survey on this
subject, see \cite{HromkovicKPRU05}.)

With the advent of mobile ad-hoc  networks there has been  growing
interest in studying information dissemination in dynamic scenarios,
where a number of agents move either in a continuous space or along
the nodes of some underlying graph and exchange information when their
positions satisfy a specified proximity constraint.
In \cite{ClementiMPS09,ClementiPS09} the authors study the time it takes
to broadcast information from one of $m$ mobile agents to all others.
The agents move on a square grid of $n$ nodes and in each time
step, an agent can (a) exchange information with all agents at distance at
most $R$ from it, and (b) move to any random node at distance at
most $\rho$ from its current position. The results in these papers 
only apply to a very dense scenario where the number of agents
is linear in the number of grid nodes
(i.e.,  $m=\bigTh{n}$).
They show that the broadcasting time is $\bigTh{\sqrt{n}/R}$ w.h.p.,
when $\rho = \bigO{R}$ and $R = \bigOm{\sqrt{\log n}}$ \cite{ClementiMPS09}, and it is
$\bigO{(\sqrt{n}/\rho)+\log n}$ w.h.p., when $\rho =
\bigOm{\max\{R,\sqrt{\log n}\}}$ \cite{ClementiPS09}.
These results crucially rely on   $R+\rho = \bigOm{\sqrt{\log n}}$,
which implies that the  range of agents' communications or  movements  
at each step defines a connected graph.

In more realistic scenarios, like the ones adopted in this paper, 
the number of agents is decoupled from
the number of locations (i.e., the graph nodes) and a smoother
dynamics is enforced by limiting agents to move only between
neighboring nodes.
A reasonable setting consists of a set of multiple, simple
random walks on a graph, one for each agent, with communication
between two agents occurring when they meet at the same node.
One variant of this setting is the so-called \emph{Frog Model}
(corresponding to our Broadcasting scenario),
where initially one  of $m$ agents is  active
(i.e., is performing a random walk), 
while the remaining agents do not move.
Whenever an active agent hits an inactive one,
the latter is activated and starts its own random walk.
This model was mostly studied in the infinite grid focusing on
the asymptotic (in time) shape of the set of vertices
containing all active agents~\cite{AlvesMP02, KestenS03}.
A model similar to our Gossiping scenario
is often employed to model the spreading of computer viruses
in networks and the spreading time is also referred to
as \emph{infection time}.
In \cite{DimitriouNS06}, the authors provide a general bound
on the average infection time when $m$ agents (one of them initially
affected by the virus) move in an $n$-node graph.
For general graphs, this bound is $\bigO{t^* \log m}$, where $t^*$ denotes the maximum
average meeting time of two random walks on the graph, and the
maximum is taken over all pairs of starting locations of the random walks.
Also, in the paper tighter bounds are provided for the complete
graph and for expanders. Observe that the  $\bigO{t^* \log m}$  bound specializes to
$\bigO{n \log n \log m}$ for the $n$-node grid by applying the known
bound on $t^*$ of \cite{AldousF98}.
A tight bound of $\bigTh{n \log n \log m /m}$ on the infection time on the grid
is claimed in \cite{WangKK08},
based on a rather informal argument where some unwarranted
independence assumptions are made.
Our results show that this latter bound is incorrect.

Finally, a related line of research deals with the cover time
of a random walk on a graph, that is, the expected time when all of the graph
nodes are touched by the random walk.
(See \cite{AldousF98} for a comprehensive account of the relevant literature.)
The cover time is strictly related to the hitting time \cite{BroderKRU94},
namely the average time required of a random walk to reach a specified node.
For $n$-node meshes, it is known that the hitting time is $\bigO{n \log n}$,
while the cover time is $\bigO{n \log^2 n}$ \cite{Zuckerman92,ChandraRRST97}.
Bounds on the speed-up achieved on the cover time by multiple random
walks as opposed to a single one are proved in \cite{AlonAKKLT08,ElsasserS09}.

\subsection{Organization of the paper}
The rest of the paper is organized as follows.
In Section~\ref{sec:prelim}, we define the problem of interest
in the two scenarios and establish some technical facts 
which are used in the analysis.
Section~\ref{sec:bs} and Section~\ref{sec:gs} contain our results
for the Broadcasting and Gossiping scenario, respectively.

\section{Preliminaries}
\label{sec:prelim}

We study the dynamics of multiple independent random walks moving
on an $n$-node 2-dimensional square grid  $G_n=(V_n,E_n)$, where
$V_n=\{(i,j)~|~1\leq i,j \leq \sqrtn\}$, and $E_n
= 
\{((i,j),(i,j+1))~|~1\leq i \leq \sqrtn, 1 \leq j < \sqrtn\}
\cup \{((i,j),(i+1,j))~|~1\leq i < \sqrtn, 1 \leq j \leq \sqrtn\}$.
We also add  self loops to the boundary nodes
$\{(1,j),(\sqrtn,j),(i,1),(i,\sqrt{n})~|~1\leq i,j \leq \sqrt{n}\}$
so to avoid bipartiteness and equalize the steady state distribution. 
We remark that all results in the paper can be immediately ported to the torus
graph, where wrap-around edges substitute the self-loops.

Although the two scenarios are defined with respect to
a fixed number $m$ of agents distributed uniformly and independently at random
among the nodes of $G_n$ (the \emph{exact model}),
technically, when deriving the upper bounds to broadcasting and gossiping times,  it is easier to work with a 
slightly modified model in which each node $v$ originally holds $m_v$ agents, where $m_v$
is distributed as a binomial variable $B(m,1/n)$ independently of the  other nodes (the \emph{binomial model}).
We denote by $\dens{}=m/n$ the \newterm{density} of the binomial model,
and by $\tilde{m}$ the random variable denoting the number of agents in a given instance of the model.
For a sufficiently large $\dens{}$, we have $\tilde{m}=\bigTh{m}$ with high probability:
\begin{lemm}
\label{lemm:ConcAgents}
Let $\dens{} = m/n \geq (17 \log n) / n$, then, with probability $1 - 1/n^2$,
the number of agents in the system, $\tilde{m}$, satisfies
 $\frac{1}{2} n \dens{} \leq \tilde{m} \leq \frac{3}{2} n \dens{}$.
\end{lemm}
\begin{proof}
The number of agents is a binomial random variable with expectation ${m} = 17 \log n$.
Applying a Chernoff--Hoeffding bound to their number,
yielding
\[
	\prob{\abs{\tilde{m} - m} \geq \frac{1}{2}m}
		\leq 2 e^{-\frac{1}{8} m}
		\leq \frac{1}{n^2},
\]
for a sufficiently large $n$.
\end{proof}
Observe that a high-probability result for the binomial model with $\dens{}=m/n$ implies
a similar high-probability result in the exact model with $m$ agents
(see~\cite[Corollary~5.9]{MitzenmacherU05}).

Moving agents follow independent, simple, symmetric random walks,
that is, at each step an agent moves to a neighbor of its current location, chosen uniformly at random. Moreover,
time is discrete, and   moves are synchronized.
For a given rumor $r$, we say that an agent is \newterm{informed} (of $r$) if it has a copy of $r$,
otherwise it is \newterm{uninformed} (of $r$). In both the broadcasting and gossiping scenarios defined in the introduction,
whenever an agent informed of $r$ meets another agent uninformed of $r$, the latter becomes informed of $r$.

\begin{defi}[broadcasting time, gossiping time]
In the \bs,
the broadcasting time $\bt$ is the first time 
at which all  agents are informed of the single rumor.
In the \gs, the gossiping time $\gt$ is the first time
at which all  agents are informed of all  rumors.
\end{defi}

We denote by $\dist{x-y}$ the Manhattan ($L_1$) distance between nodes $x$ and $y$ of $G_n$.
Our analysis uses the following two technical lemmas which characterize
the set of nodes visited by a random walk within a given time interval.
\begin{lemm}
\label{lemm:SRW}
Consider a random walk on $G_n$, starting at time $t=0$ at node $v_0$.
There exists a positive constant $c_1$ such that
for any $v \neq v_0$,
\[
	\prob{v \mbox{ is visited within } (\dist{v-v_0})^2 \mbox{ steps}}
		\geq \frac{c_1}{\max \{1,\log (\dist{v-v_0})\}}.
\]
\end{lemm}
\begin{proof}
The Lemma  is proven in \cite[Theorem~2.2]{AlvesMP02}
for the infinite grid $\Z^2$. By the ``Reflection Principle''~\cite[Page 72]{Feller68},
for each walk in $\mathbb{Z}^2$ that started in $G_n$,
crossed a boundary and then crossed the boundary back to $G_n$,
there is a walk with the same probability that does not cross the boundary
and visits all the nodes in $G_n$ that were visited by the first walk.
Thus, restricting the walks to $G_n$ can only change the bound by a constant factor.
\end{proof}

\begin{lemm}
\label{lemm:props}
Consider the first $\ell$ steps of a random walk in $G_n$
which was at node $v_0$ at time $0$.
\begin{enumerate}
	\item\label{poin:dev} The probability that at any given step $1\leq i\leq \ell$
		the random walk is at distance at least $\geq \lambda\sqrt{\ell}$
		from $v_0$ is at most $2 e^{-\lambda^2/2}$.
	\item\label{poin:range} There is a constant $c_2$ such that, with probability greater than $1/2$,
		by time $\ell$ the walk has visited at least $c_2\ell/\log {\ell}$ distinct nodes in $G_n$.
\end{enumerate}
\end{lemm}
\begin{proof}
We observe that the distance from $v_0$
in each coordinate defines a martingale with bounded difference $1$.
Then, the first property follows from  the Azuma-Hoeffding Inequality~\cite[Theorem 2.6]{MitzenmacherU05}.

As for the second property, let $R_\ell$ be the set of nodes reached by the walk in $\ell$ steps.
By Lemma~\ref{lemm:SRW}, $\expe{R_\ell}=\bigOm{{\ell}/{\log \ell}}$ (even when $v_0$ is near a boundary),
while $\vari{R_\ell}=\bigTh{{\ell^2}/{\log^4 \ell}}$ (see~\cite{Torney86}). The result follows by
applying Chebyshev's inequality.
\end{proof}

\section{The \bs} \label{sec:bs}
In the following two subsections we derive upper and lower bounds on $\bt$,
the first time at which all agents are informed of the single rumor.

\subsection{An upper bound on the broadcasting time}
Since the cover time of $G_n$ is $\bigO{n\log^2 n}$ (see \cite{Zuckerman92, ChandraRRST97}),
we can easily achieve an  $\bigOt{n/\sqrt{m}}$ upper bound on $\bt$, with high probability, when $m$ is 
polylogarithmic in $n$. Therefore, in what follows we focus on the case $m=\bigOm{\log^3 n}$.

To ease the analysis, it is convenient to envision the spreading process as partitioned into 
three successive phases (see Figure~\ref{figu:BroadcastingScenario};
the value $\ell_1$ and the constants $q, q'$ will be set by the analysis):
\begin{description}
	\item[Phase~I] (\emph{Initial diffusion})
		The source agent informs a set $A_1$ of at least $q \log^2 n$ agents,
		all with origins in a square $B_1$ of side length $8\ell_1 \log^{3/2} n$.
	\item[Phase~II] (\emph{Covering $B_1$})
		All agents originally placed in $B_1$ are informed.
	\item[Phase~III] (\emph{Covering $G_n$})
		All remaining uninformed agents are informed.
\end{description}

\begin{figure}[h]
\centering
\includegraphics[width=0.5\textwidth]{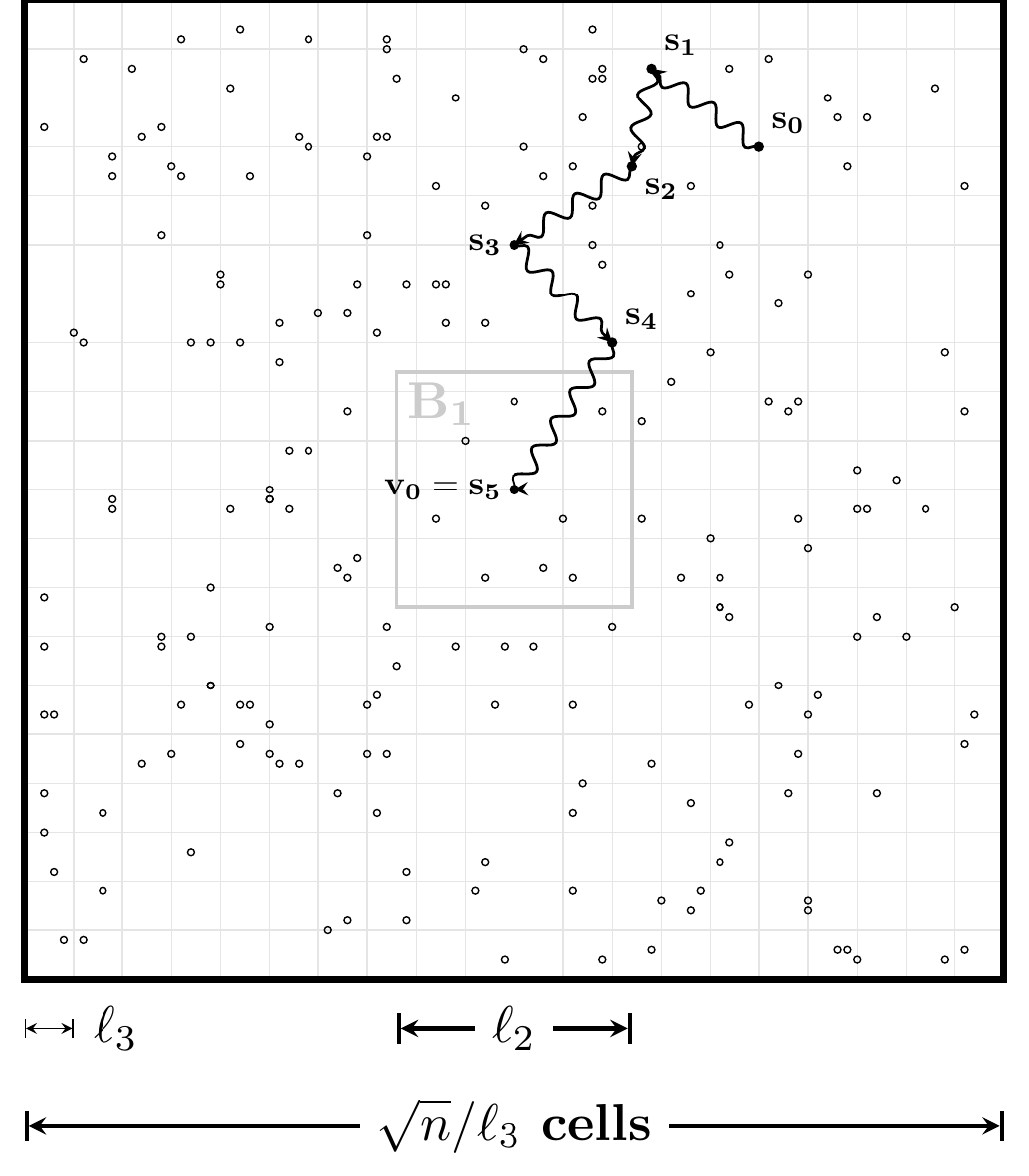}
\caption{\bs.
	In the figure, at the beginning the source walks for $5$ intervals of length $\ell_1^2$
	before reaching point $v_0$ at time $\bar{t} = 5\ell_1^2$.
	We imagine that in the interval starting at $\bar{t}$ the source visits $\bigOm{\ell_1^2/\log \ell_1}$
	distinct nodes, each inside the square $B_1$ of side $\ell_2$.
	The figure also shows the tessellation with side $\ell_3$ adopted when analyzing the covering of $G_n$.}
\label{figu:BroadcastingScenario}
\end{figure}

The following lemma bounds from above the completion time of Phase~I.
\begin{lemm}
\label{lemm:FirstPhase}
Let $\ell_1 = \sqrt{{(4 q \log^3 n})/({c_2 \dens{})}}$,
for constants $c_2>0$ (defined in Lemma~\ref{lemm:props})
and $q > 0$.
Let $T_1 = 3 \ell_1^2 \log n$. With high probability,  at time $T_1$ 
there is a set $A_1$ of informed agents  such that:
\begin{enumerate}
	\item $\card{A_1} \geq q \log^2 n$;
	\item for each pair of agents $a_1,a_2 \in A_1$,
		their initial positions at $t=0$
		are within distance $8 \ell_1 \log^{3/2} n$.
\end{enumerate}
\end{lemm}
\begin{proof}
Partition the $T_1$ time steps into
$3 \log n$ disjoint intervals of length $\ell_1^2$,
and consider the path of the source agent during one such interval.
Let $v_0$ be the location of the agent at the beginning of the interval.

By Lemma~\ref{lemm:props}, with probability $2e^{-\log^3 n} \leq 1/n^2$
all the nodes visited by the walk are within distance $4 \ell_1 \log^{3/2} n$ from $v_0$,
and thus within distance $8 \ell_1 \log^{3/2} n$ of one  other.
Moreover, the same lemma implies that
with probability $\geq 1/2$ at least $c_2 \ell^2_1/\log \ell_1$ distinct nodes are visited.

Conditioning on visiting these many  nodes in the interval,
the expected number of agents informed by the walk
is $c_2 ({\ell^2_1}/{\log \ell_1})\dens{} \geq 2q\log^2 n$.
Applying the Chernoff bound, with probability at least $1-1/n$,
the number of agents informed by the path is at least $q\log^2 n$.

Thus, with probability at least $1/2-2/n$ one segment of the path satisfies the claim,
hence the probability that the claim  holds for one of the $3 \log n$ segments is at least  $1-1/n$ for $n$ large enough.
\end{proof}

Consider now an arbitrary square area $B_1$
with side $\ell_2 = 8 \ell_1 \log^{3/2} n$,
containing all  initial locations of the agents  of the set $A_1$ of Lemma~\ref{lemm:FirstPhase}.
Next lemma bounds the time taken to complete Phase~II.
\begin{lemm}
\label{lemm:SecondPhase}
Let $T_2 = 4 \ell_2^2$.
With high probability, every agent initially located in  $B_1$ at $t = 0$
is informed by time $\tau_2 = T_1 + T_2$.
\end{lemm}
\begin{proof}
We show that every node  in $B_1$ is  visited by at least one agent in $A_1$ before time $T_2$.

There are $\ell_2^2$ nodes in $B_1$, and the diameter of  $B_1$ is at most $2\ell_2$.
Applying Lemma~\ref{lemm:SRW} and the fact that the $\card{A_1}$ walks are independent,
the probability that any node in $B_1$
is not visited by some agent in $A_1$ until  time $T_2$ is at most 
\[
	\ell_2^2 \left(1 - \frac{c_1}{\log (2\ell_2)}\right)^{\card{A_1}}
		\leq \exp\left(-\frac{c_1 q \log^2 n}{\log (2\ell_2)} +2 \log \ell_2\right)
		\leq \frac{1}{n},
\]
by selecting a suitably large value for the positive constant $q$.
\end{proof}

To analyze the completion time of Phase~III, we consider a tessellation of $G_n$ into cells
of side $\ell_3 = \sqrt{{q' \log^3 n}/{\dens{}}}$,
for some suitable constant $q' > 0$, such that
one of the cells is entirely contained in $B_1$.
We show that once all agents placed in one cell are informed,
then with high probability all  agents placed in the adjacent cells
are informed within $20\ell_3^2$ steps.
\begin{lemm}
\label{lemm:ThirdPhase}
Consider a cell $C$, and let $\tau$ be the first time when all the agents originating in $C$ are informed.
With probability greater than $1-1/n^2$, all  agents originating  in the (at most)
$4$ cells adjacent to $C$ are informed at time $\tau+20\ell^2_3$.
\end{lemm}
\begin{proof}
By the Chernoff bound, at least $\ell^2_3 \dens{}/2 = {q'}\log^3 n / 2$  agents
originate from a cell $C$, with probability at least $1-1/2n^2$.
The maximum distance between a node in $C$ and a node in an adjacent cell
is bounded by $3\ell_3$.
Thus, the probability that any node in the adjacent cells
is not visited by one of the agents placed in $C$ within $20\ell^2_3$ steps
after time $\tau$ is bounded by
\[
	4\ell^2_3 \left(1 - \frac{c_1}{\log (3\ell_3)}\right)^{q'\log^3 n/2}
		\leq \exp\left(-\frac{c_1 q' \log^2 n}{2 \log (3\ell_3)} + 2 \log 2\ell_3\right)
		\leq \frac{1}{n^3},
\]
by selecting a suitably large value for the  constant $q'$.
\end{proof}

We are now ready to prove the main result of the subsection:
\begin{theo}
\label{theo:UBSpreadingTime}
With high probability,
\[
	\bt = \bigOt{\frac{n}{\sqrt{m}}}.
\]
\end{theo}
\begin{proof}
The case $m=\bigO{\log^3n}$ immediately follows from the observation made at the beginning of the subsection.
Consider now the case $m=\bigOm{\log^3n}$.

With probability $1-1/n$, Phase~I terminates in $\ell^2_1 \log n$ steps (Lemma~\ref{lemm:FirstPhase})
and with probability  $1-1/n$, Phase~II terminates in $4\ell^2_2$ steps (Lemma~\ref{lemm:SecondPhase}).

In Phase~III, with probability $1-1/n^2$, all agents in a given cell are informed no later than $20\ell^2_3$ steps
after all  agents in one of its adjacent cells are informed.
Since the ``cell distance'' (number of cells) between any two cells is at most  $2\sqrt{n}/\ell_3$,
Phase~III terminates in ${(2\sqrt{n}}/{\ell_3}) 20\ell^2_3 = 40\ell_3\sqrt{n}$ steps
with probability at least $1-1/n$.

Putting it all together, we obtain that, with probability at least $1-\bigO{1/n}$,
\[
	\bt
		\leq 3 \ell^2_1 \log n + 4 \ell^2_2 +40\ell_3\sqrt{n}
		= \bigO{\frac{n\log^{3/2} n}{\sqrt{m}}}.\qedhere
\]
\end{proof}

\subsection{A lower bound on the broadcasting time}
\label{sect:LowerBound1}
We develop a lower bound of $\bigOmt{n / \sqrt{m}}$ to the broadcasting time.
The argument relies on the intuition that many agents are initially located
at a large distance from one another: this limits the rate at which the rumor may spread. 
We need the following definition (see also Figure~\ref{figu:Island}):
\begin{defi}[Island]
\label{defi:Island}
Let $A$ be the set of agents.
For any parameter $\gamma > 0$,
let $G_t(\gamma)$ be the graph  with vertex set $A$ 
and such that there is an edge between two vertices
iff the corresponding agents  are within distance $\gamma$ at time $t$.
The \newterm{island} of parameter $\gamma$ of an agent $a$ at time $t$,
denoted by $\isla{a}{\gamma}{t}$, is the connected component of $G_t(\gamma)$ containing $a$.
\end{defi}

\begin{figure}[h]
\centering
\includegraphics[width=0.5\textwidth]{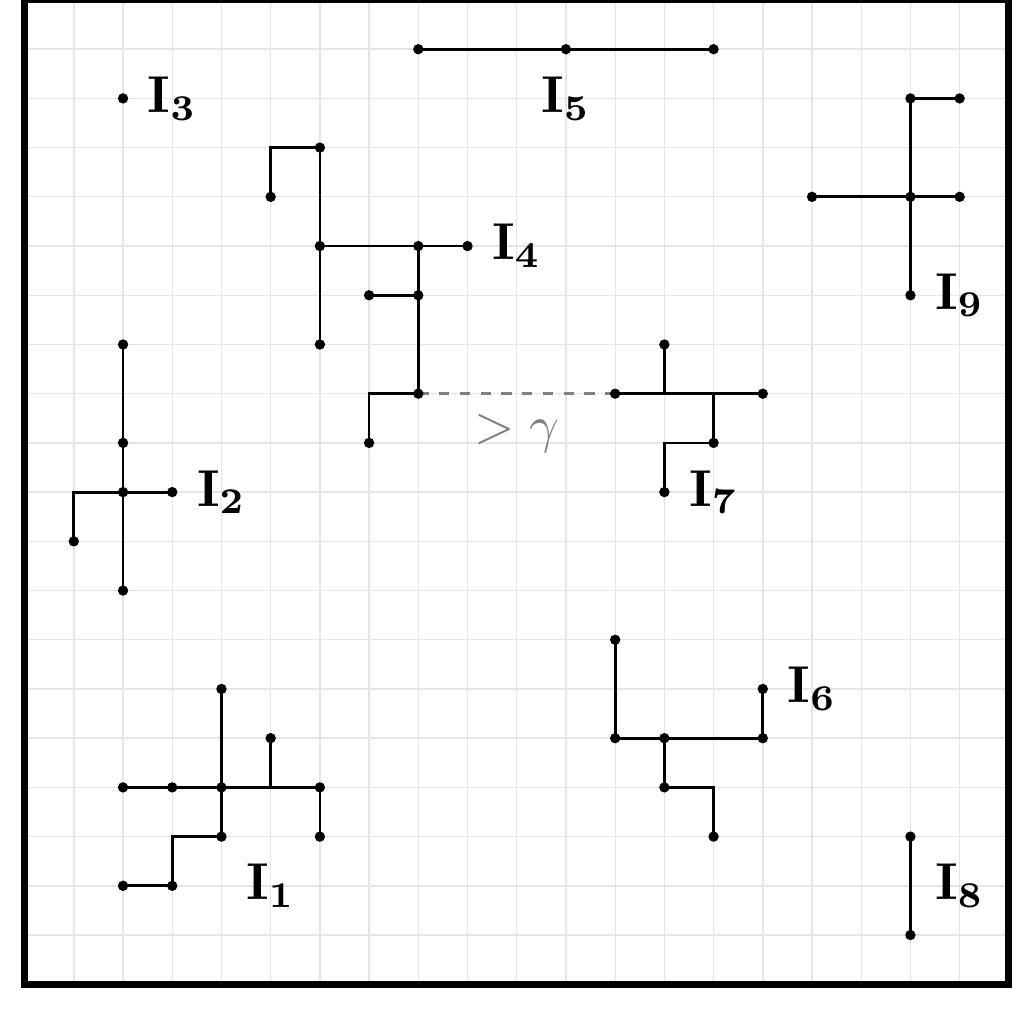}
\caption{A partition of agents into islands with $\gamma = 3$.}
\label{figu:Island}
\end{figure}

In this section, we only consider the islands at time $t=0$.
Islands at $t>0$ are used for the lower bound on the
gossiping time in the next section.
We first prove that with high probability there are no big islands.
\begin{lemm}
\label{lemm:NoBigIslands}
Let $\gamma = \sqrt{n/(4 e^3 m)}$  and
write $A= \cup_{j=1}^{s} I_j$ as the disjoint union of $s$ islands of parameter $\gamma$.
Then, with high probability, $\card{I_j} \leq \log n$, for all $1 \leq j \leq s$.
\end{lemm}
\begin{proof}
Let $\B_k$ denote the event that there exists an island with at least $k$ agents.
It is easy to see that $\prob{\B_k}$ is upper bounded by the probability that $G_t(\gamma)$
contains a tree of $k$ vertices of  $A$ as a subgraph. Since $k^{k-2}$ is the number
of unrooted labeled trees on $k$ nodes, and ${4 \gamma^2}/{n}$ is (an upper bound to)
the probability that a given agent lies within distance $\gamma$ from another given agent, we have that 
\[
	\prob{\B_k}
		\leq \binom{|A|}{k} k^{k-2} \left(\frac{4 \gamma^2}{n}\right)^{k-1}
		\leq \left(\frac{e|A|}{k}\right)^k k^{k-2} \left(\frac{4 \gamma^2}{n}\right)^{k-1}.
\]
Since $|A|=m$, setting $k =1 + \log n$ and substituting the definition of $\gamma$, we obtain
\[
	\prob{\B_k} \leq \frac{e|A|}{k^2} e^{-2(k-1)}\leq \frac{e m}{k^2} \frac{1}{n^2}
		\leq \frac{1}{n}.\qedhere
\]
\end{proof}

Since the agents are distributed uniformly and independently at random,
with high probability there exists an agent placed
at distance $\bigOm{\sqrt{n}}$ from the source of the rumor:
\begin{lemm}
\label{lemm:Distance}
With probability $1-1/n^2$, at time $0$ at least one agent
is placed at $L_1$ distance $\geq \sqrt{n}/2$ from the source agent.
\end{lemm}
\begin{proof}
Given node $v$ of a grid, there are up to $4i$ nodes at distance $i$ from $v$,
so there are up to $n/2$ nodes at distance $<\sqrt{n}/2$ from the source agent.
The expected number of agents outside this area is $n \dens{} / 2 = \bigOm{\log^3 n}$,
thus the probability that no agent is placed at distance at least $\sqrt{n}/2$
from the source agent is bounded by $1/n^2$.
\end{proof}

We are now ready to prove the lower bound on the broadcasting time:
\begin{theo}
\label{theo:LBSpreadingTime}
With high probability,
\[
	\bt = \bigOm{\frac{n}{\sqrt{m}\log^2 n}}.
\]
\end{theo}
\begin{proof}
Let $v_0$ denote the agent placed at distance
at least $\sqrt{n}/2$, whose existence is guaranteed
with probability $1 - 1/n^2$ by Lemma~\ref{lemm:Distance}.

By setting  the parameter $\gamma$ as in Lemma~\ref{lemm:NoBigIslands}, with high probability 
no island has more than $\log n$ agents, hence
the maximum distance between two agents in the same island is at most $\gamma(\log n - 1)$.
Thus, with high probability, the rumor must traverse at least $H=\sqrt{n}/(2\gamma(\log n -1))$ islands to reach node 
$v_0$.

\sloppy{%
Hence, for $H-1$ non-overlapping time periods, the rumor must cover the inter-island distance which is at least $\gamma$.
By applying Lemma~\ref{lemm:props}, the probability that one inter-island distance be covered in at most
$\tau = (1/(24e^3))n/(m\log n)$ steps is at most $1/n^2$, therefore, with probability at least $1-1/n$
the time to spread the  rumor is at least $(H-1)\tau = \bigOm{n/(\sqrt{m}\log^2n)}$.}
\end{proof}

\section{The \gs}\label{sec:gs}

The \gs\ differs from the \bs\ in that each agent has initially a distinct rumor to spread,
and all agents perform independent random walks starting at time 0.

\subsection{The meeting probability of two random walks}
The main new ingredient in analyzing the \gs, which is also a result of independent interest, 
is a lower bound on the probability that two random walks on the grid meet within a given time interval.
The following lemma is the analogous of Lemma~\ref{lemm:SRW} for the case of two walks.
The proof, however, requires a different argument.

\begin{lemm}
\label{lemm:MeetingProbability}
Consider two independent, simple random walks $\bar{a}$ and $\bar{b}$,
starting at time $0$ at node $a_0$ and $b_0$, respectively, with $a_0\neq b_0$.
Let $a_t$ and $b_t$ be the locations of the walks at time $t$ and
let $T\geq \dist{a_0 -b_0}^2$.
Then, there exists a constant $c_3>0$ such that
\[
	P_{\bar{a},\bar{b}}(T) \triangleq \prob{\exists t\leq T~\mbox{such that}~a_t=b_t} \geq \frac{c_3}{\max\{1, \log (\dist{a_0 -b_0})\}}.
\]
\end{lemm}
\begin{proof}
The case $\dist{a_0-b_0}=1$ is immediate. Consider now the case $\dist{a_0-b_0}>1$.
Let $P_t (w, x)$ denote he probability that a walk that started at node $w$ at time 0 is at node $x$ at time $t$, and 
let $R(w,u,s)$ be the expected number of times that two walks
which started at nodes $w$ and $u$ at time 0 meet during the time interval $[0,s]$, then
\[
	R(w,u,s)=\sum_{t=0}^s \sum_x P_t (w, x)P_t (u,x).
\]
Let $\tau (a,b)$ be the first meeting time of the walks $\bar{a}$ and $\bar{b}$. Then
\[
	R(a_0,b_0,T) = \sum_{t=0}^T \prob{\tau (a,b)=t}R(x_t,x_t,T-t) \leq P_{\bar{a},\bar{b}}(T) \max_x R(x,x,T).
\]
Thus, setting $T_0 = \dist{a_0 -b_0}^2$, we have
\[
	P_{\bar{a},\bar{b}}(T) \geq P_{\bar{a},\bar{b}}(T_0) \geq \frac{R(a_0,b_0,T_0)}{\max_x R(x,x,T_0)}.
\]

Let $D(a_0,b_0)$ be the set of nodes with distance up to $\dist{a_0-b_0}$ from both $a_0$ and $b_0$, i.e.,
\[
	D(a_0,b_0)=\{x~|~\dist{x-a_0}\leq \dist{a_0-b_0}~\mbox{and}~\dist{x-b_0}\leq \dist{a_0-b_0}\}.
\]
It is easy to verify that $\card{D(a_0,b_0)} \geq \frac{1}{4} \dist{a_0-b_0}^2$.
Applying Theorem 1.2.1 in \cite{Lawler91} we have:
\begin{eqnarray*}
R(a_0,b_0,T_0)
	&\geq &\sum_{t=0}^{T_0} \sum_{x\in D(a_0,b_0)} P_t (a_0, x)P_t (b_0,x)\\
	&\geq &\sum_{t=\frac{T_0}{2}+1}^{T_0} \sum_{x\in D(a_0,b_0)} 4\left(\frac{1}{\pi t}\right)^2e^{-\frac{\dist{x-a_0}^2 + \dist{x-b_0}^2}{t}}.
\end{eqnarray*}
By bounding $\dist{x-a_0}^2$ and  $\dist{x-b_0}^2$ from above with $T_0$ in the  formula, easy calculations show  that 
$R(a_0,b_0,T_0) = \bigOm{1}$.
Similarly, using the fact that there are no more than $4i$ nodes at distance exactly 
$i$ from $x$, we have:
\begin{eqnarray*}
\max_x R(x,x,T) &\leq& 1 + \sum_{t=1}^T \sum_{i = 1}^{t} 4i\, 4\left(\frac{1}{\pi t}\right)^22 e^{-\frac{i^2}{t}} \\
&\leq & 1 + \left(\frac{4}{\pi}\right)^2 \sum_{t=1}^T \frac{1}{t^2} \left(\left(\sum_{i=1}^{\sqrt{t}} i\right) + \left(\sum_{i=1+\sqrt{t}}^{t} i e^{-i^2/t}\right)\right)\\
&\leq & 1 + \left(\frac{4}{\pi}\right)^2 \sum_{t=1}^T \frac{1}{t^2} \left(\frac{t}{2} + \left(\sum_{i=1+\sqrt{t}}^{t} i^2 e^{-i^2/t}\right)\right)\\
&\leq & 1 + \left(\frac{4}{\pi}\right)^2 \sum_{t=1}^T \frac{1}{t^2} \left(\frac{t}{2} + \frac{e}{(e-1)^2} t \right) = \bigO{\log T}.
\end{eqnarray*}

We conclude that there is a constant $c_3>0$ such that
\[
P_{\bar{a},\bar{b}}(T) \geq \frac{R(a_0,b_0,T_0)}{\max_x R(x,x,T_0)}\geq \frac{c_3}{\log (\|a_0 -b_0\|)}.\qedhere
\]
\end{proof}

\subsection{An upper bound on the gossiping time}
We observe that since the $L_1$ diameter of $G_n$ is $2\sqrt{n} - 2$,
we can use Lemma~\ref{lemm:MeetingProbability} to show that
with probability $1-1/n^2$, at time $8 n \log^2 n$
an agent  has met all other agents walking   in $G_n$.
Thus, the theorem trivially holds for $m$ polylogarithmic in $n$.

Consider now the  case $m=\bigOm{\log^3 n}$.
As in  the derivation of the upper bound to the Broadcasting time,
we resort to the binomial model introduced in Section~\ref{sec:prelim}.
Let $\tilde{m}$ be the number of agents in a given instance of the process and recall that $\tilde{m}=\bigTh{m}$ w.h.p.

We prove the theorem by bounding from above the spreading time of any fixed rumor.
Specifically, let $\bt^i $ be the time for spreading of the  $i$-th rumor.
We show that  there is a constant $c_4>0$ such that, for $1 \leq i \leq \tilde{m}$,
\[
	\prob{\bt^i \geq \frac{c_4 n\log^{3/2} n}{\sqrt{{m}}}} \leq \frac{1}{n^2}.
\]
The argument proceeds as follows. We tessellate $G_n$ 
into cells of suitable side $\ell_1$ (defined in Lemma~\ref{lemm:FirstPhase2}) and  say  that a cell $Q$ is \newterm{reached}
at time $t_Q$ if $t_Q$ is the first time when a node of the cell
hosts an  agent informed of the rumor.  We call this first visitor the \newterm{explorer} of $Q$, and 
let $A_Q$ denote the set of agents inside $Q$ at time $t_Q$.

We first show that, after a suitably chosen number of steps after  $t_Q$, cell $Q$ is \newterm{conquered},
in the sense that a large number of agents of $A_Q$ have been informed.
Also, we show that by the time  cell $Q$ is conquered, each of its neighboring cells has been reached,
and that a large number of informed agents are within a short distance from $Q$.
These facts imply that the conquering process proceeds smoothly. Finally, we prove that
all agents are informed of the rumor shortly after all cells are conquered.
This argument is made rigorous in the following sequence of lemmas.
\begin{lemm}
\label{lemm:FirstPhase2}
Let $\ell_1 = \sqrt{{4 qn \log^3 n}/{(c_3 m){}}}$, 
where $q > 0$ and  $c_3$ is defined in Lemma~\ref{lemm:MeetingProbability}. Let also
 $T_1 = 4 \ell_1^2$. Then, for any cell $Q$ of the tessellation,
 by time $\tau_1 = t_Q + T_1$,
at least $q \log^2 n$ agents of $A_Q$ are informed  with probability $1-{1/n^6}$, for sufficiently large $n$.
\end{lemm}
\begin{proof}
Since at any given time the agents are at random and independent locations, 
the expected number of agents in $Q$ at time $t_Q$
is $ {4 q \log^3 n}/{c_3}$.
Thus, by the Chernoff bound  at least ${4 q \log^3 n}/{(2c_3)}$
agents are inside $Q$ at time $t_Q$  with probability at least $1-1/n^6$.

For each agent $a \in A_Q$, define an indicator variable $X_a$ such that $X_a =1$ iff agent $a$
has met the explorer of $Q$ by time $\tau_1$.
The $X_a$'s are independent and
$X = \sum_{a \in A_Q} X_a$ is a lower bound to the number of informed agents in $A_Q$
by time $\tau_1$. Applying Lemma~\ref{lemm:MeetingProbability}, we have 
$\prob{X_a = 1} \geq \frac{c_3}{\log 2 \ell_1}$, hence $\expe{X} \geq  2 q \log^2 n$. 
Using the Chernoff bound, we finally obtain
$\prob{X < q \log^2 n} \leq \exp\left(- {q \log^2 n}/{4} \right) \leq {1}/{n^6}$,
for a sufficiently large $n$.
\end{proof}

\begin{lemm}
\label{lemm:SecondPhase2}
Consider a cell $Q$, and 
define $R$ to be the square of side $\ell_2 = (1+2\sqrt{2})\ell_1$ centered at $Q$
(see Figure~\ref{figu:GossipingScenario}). Let $\tau_1$ and $T_1$ 
be defined as in Lemma~\ref{lemm:FirstPhase2},  and let $T_2 = 4 \ell_2^2$.
Let $A_R \subseteq A_Q$ denote the set of agents residing inside $Q$ at time $t_Q$
and contained in $R$ at time $\tau_1$. Finally, let $\tau_2 = \tau_1 + T_2 = t_Q + T_1 + T_2$.
Then, with probability at least $1 - 1/n^5$,
\begin{enumerate}
	\item\label{poin:branch} at least $\frac{1}{2} (1-\frac{2}{e}) q \log^2 n$ informed agents
		are within distance $\ell_2$ from $Q$ at time $\tau_1$; 
	\item $\frac{1}{4} (1-\frac{2}{e}) \frac{4q \log^3 n}{c_3} \leq \card{A_R} \leq \frac{9}{4} (1-\frac{2}{e}) \frac{4q \log^3 n}{c_3}$; 
	\item\label{poin:capture} every agent $a \in A_R$ is informed by time $\tau_2$; and
	\item\label{poin:neighbor} each of the neighboring cells of $Q$ has been reached by time $\tau_2$.
\end{enumerate}
\end{lemm}
\begin{proof}
Observe that by Lemma~\ref{lemm:props},
the probability that an agent $a$ 
contained in $Q$ at time $t_Q$
is outside $R$ at time $\tau_1 = t_Q + T_1$ is at most $2/e$, whence
$ \prob{a \in A_R | a\in A_Q} \geq \left(1 - \frac{2}{e}\right) \triangleq r$.
Since (by Lemma~\ref{lemm:FirstPhase2}) at least $q \log^2 n$  agents of $A_Q$ are informed at time $\tau_1$, and
each of them belongs to $A_R$  independently with probability $r$, we can apply 
the Chernoff-Hoeffding bound to show that with probability $1-1/n^6$ the number of informed agents in $R$ at time $\tau_1$
is at least $I_R(\tau_1) = \frac{r}{2} q \log^2 n$.
This proves Point (1), and Point (2) follows by a similar argument.

For what concerns Point (3),
consider an agent $a \in A_R$, not informed at time $\tau_1$.
Since $a$ is within distance $2\ell_2$ from the informed nodes belonging to $A_R$,
by Lemma~\ref{lemm:MeetingProbability},
the probability that $a$ is not informed at time $\tau_2$ is at most
\[
	\left(1 - \frac{c_3}{\log 2\ell_2}\right)^{I_R(\tau_1)}
		\leq \exp\left(- \frac{c_3 I_R(\tau_1)}{\log 2\ell_2}\right)
		< \frac{1}{n^6}
\]
by selecting a suitably large constant $q$.  The point is proven by applying the union bound on
$\card{A_R} = \bigO{\log^3 n}$ agents.

Finally, for Point (4), consider one of the neighboring cells of $Q$, say $Q'$.
Since $Q' \subseteq R$, each point of $Q$ is at distance at most $2 \ell_2$
from the actual location of the $I_R(\tau_1)$ informed agents
which are inside $R$ at time $\tau_1$.
Therefore, applying Lemma~\ref{lemm:SRW},
we can conclude that the probability that $Q'$ is not reached
by time $\tau_2$ is at most
\[
		\left(1 - \frac{c_1}{\log 2\ell_2}\right)^{\card{Q'} I_R(\tau_1)}
		\leq \exp\left(-\frac{c_1}{\log 2\ell_2} \frac{r q \log^2 n}{2} \frac{4 q \log^3 n}{c_3 m/n} \right)
		< \frac{1}{n^6},
\]
for a sufficiently large $n$.
We conclude the proof by applying the union bound over the neighboring cells.
\end{proof}

\begin{figure}[h]
\centering
\includegraphics[width=0.5\textwidth]{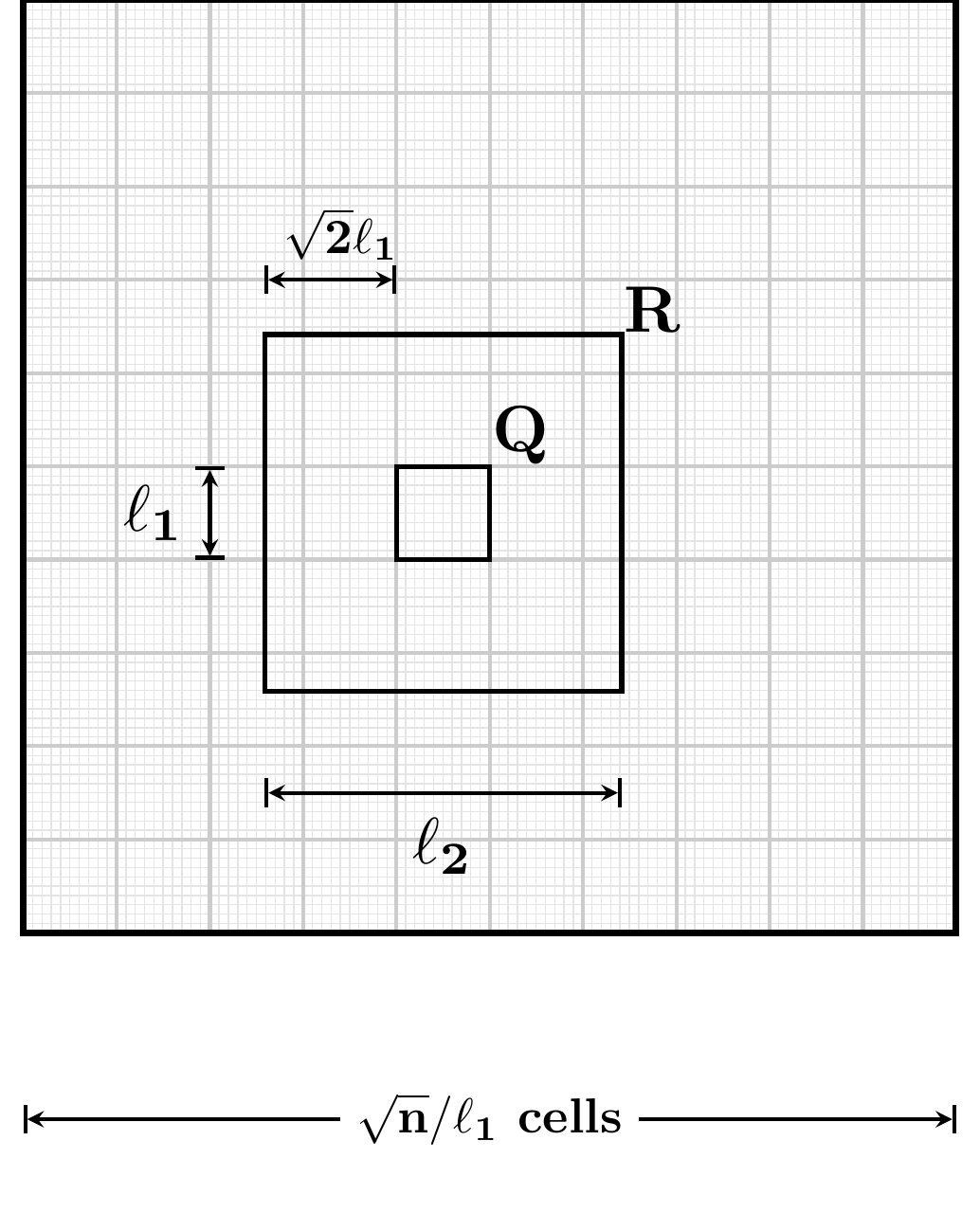}
\caption{\gs. The tessellation adopted in the analysis of the spreading of a rumor.
	A cell $Q$ and the corresponding area $R$ around it are shown.}
\label{figu:GossipingScenario}
\end{figure}

A cell $Q$ is conquered when all agents in the set $A_R$, defined in the previous lemma, are informed. 
\begin{lemm}
\label{lemm:ConqCell}
With the same notation of Lemma~\ref{lemm:SecondPhase2},
with probability $1 - 1/n^3$, for any time instant $\tau_2 \leq t \leq 8 n \log^2 n$,
there are at least $\bigOm{\log^3 n}$ informed agents
at distance at most $\ell_2$ from any node of $Q$.
\end{lemm}
\begin{proof}
Consider an agent $a \in A_R$ and let $I_Q(t)$ denote the set of informed agents
inside $Q$ at time $t \geq \tau_2$.

By Lemma~\ref{lemm:SecondPhase2}, the number $N_R(\tau_1)$ of agents inside $R$ at time $\tau_1$ is
at least $|A_R| = \bigOm{\log^3 n}$.
Since each of the agents of $A_R$ was inside cell $Q$ at time $t_Q$,
the probability $r'$ that $a$ is still inside $Q$ at time $\tau_2$ is
$ \bigTh{\ell_1^2 / \ell_2^2} = \bigTh{1}$.
Since the agents perform independent random walks,
we can apply a concentration bound and get that
\[
	\card{I_Q(\tau_2)} \geq \frac{1}{2} r' N_R(\tau_1) \geq \frac{1}{8} \left(1-\frac{2}{e}\right) r' \frac{4q \log^3 n}{c_3} \tag{1}
\]
with probability greater than $1 - 1 / n^5$ for a sufficiently large $n$.
Next, observe that the probability that one of the agents  of $I_Q(\tau_2)$
is still inside $R$ (i.e., is within distance $\ell_2$ from $Q$)
at time $\tau_3 = \tau_2 + T_1 + T_2$
is at least $r = (1-\frac{2}{e})$ by Lemma~\ref{lemm:props}.
Therefore, applying again the concentration bound,
we have that,
with probability greater than $1 - 1 / n^5$ (for a sufficiently large $n$),
there are at least
\[
	N_R(t) \geq \frac{r' r^2}{4 c_2} q \log^3 n \tag{2}
\]
informed nodes within distance $\ell_2$ from $Q$
in each of the time instants $\tau_2 \leq t \leq \tau_3$.

Now consider the time interval $[\tau_2, 8 n \log^2 n]$
and partition it into consecutive, non-overlapping subintervals of length $T_1 + T_2$.
Let $t_k = \tau_2 + (k-1)(T_1 + T_2)$ denote the beginning of the $k$-th subinterval.
For any subinterval $k$, we show that
(a) $\card{I_Q(t_k)} \geq 1$ and
(b) at least $\bigOm{\log^3 n}$ informed agents are inside $R$
during the $k$-th subinterval.
Note that, by (1) and (2),
the two properties hold for the first subinterval $[\tau_2, \tau_3)$
with probability at least $1 - 1/n^5$.

The key idea for extending the properties to $k>1$ is that
we can pick an arbitrary informed agent $a$ inside $Q$
at time $t_k$ and make it play the role of the explorer of $Q$.
More precisely, by mimicking the analyses of Lemmas~\ref{lemm:FirstPhase2} and~\ref{lemm:SecondPhase2},
provided that we replace $t_Q$ with $t_k$, 
we can show that the presence of the explorer ensures the existence of $\bigOm{\log^3 n}$ 
informed agents at the beginning of the next interval. 
This observation, combined with (1) and (2), proves that if properties (a)-(b) hold
for the $k$-th subinterval, they will hold also for the $(k+1)$-th subinterval.

Applying the union bound over $\bigO{n \log^2 n}$ subintervals concludes the proof.
\end{proof}

We are now ready to prove the main theorem of this subsection:
\begin{theo}
\label{theo:UBSpreadingTime2}
With high probability,
\[
	\gt = \bigOt{\frac{n}{\sqrt{m}}}.
\]
\end{theo}
\begin{proof}
As observed at the beginning of the subsection,
we can limit ourselves to the case $m=\bigOm{\log^3 n}$.
Consider first the spreading of the $i$-th rumor, and 
tessellate the domain with cells of side $\ell_1$, specified  in Lemma~\ref{lemm:FirstPhase2}.

Consider a cell $Q$ reached for the first time at time $t_Q$.
By Lemma~\ref{lemm:SecondPhase2},
we know that each of the neighboring cells of $Q$
is reached within time $\tau_2 = t_Q + T_1 + T_2$
with probability $1 - 1/n^5$.
Therefore, all  cells are reached within time
$({2\sqrt{n}}/{\ell_1})(T_1 + T_2)$
with probability $1 - 1/n^4$.

For any agent $a$, let $t_a$ be the first time when the agent is inside a cell
that is either reached or conquered at time $t_a$, and observe that $t_a \leq (2\sqrt{n}/\ell_1 +1)(T_1+T_2)$
with probability $1 - 1/n^4$.
Using the result of Lemma~\ref{lemm:ConqCell} about the number of informed agents near the cell
where agent $a$ resides at time $t_a$,
an argument similar to the one used in the proof of Point (3) of Lemma~\ref{lemm:SecondPhase2}
shows that $a$ will be informed by time $t_a + T_2$ with probability $1 - 1/n^3$.

Putting it all together,
by setting $c_4 = 64 (5+2\sqrt{2}) \sqrt{q/c_3}$, we have proved that
the broadcasting time $\bt^i$ of the $i$-th rumor satisfies
\[
	\prob{\bt^i \geq \frac{c_4 n \log^{3/2} n}{\sqrt{m}}} \leq \frac{1}{n^2}.
\]

The theorem follows by observing that the broadcasting time of a rumor
does not depend on the origins of the other rumors, so
we can take the union bound over all $m = \bigO{n}$ different rumors
and conclude that $\gt = \bigO{n \log^{3/2} n / \sqrt{m}}$
with probability $1 - 1/n$.
\end{proof}

\subsection{A lower bound on the gossiping time}
In the Gossiping scenario uninformed agents move, hence we need to resort to a lower bound
argument which takes such movements into account. We define the \emph{informed area} $\I(t)$ at time $t$
as the set of places visited by an informed agent up to time $t$.
The \newterm{frontier} of $\I(t)$ is the border separating
the informed area from the remaining places.
To simplify the exposition,
we imagine that the rumor travels from left to right and
consider the grid node $\point{x}(t)$
of the rightmost informed agent at time $t$.
By definition, the informed area lies to  the left of $\point{x}(t)$ and
we need to show that there is  a sufficiently large value $T$ such that, at time $T$,
there is at least one uniformed agent right of $\point{x}(T)$.

First, we can prove that,
with probability $1 - 1/n^2$,
at each time instant  $0 \leq t \leq 8 n \log^2 n$,
every island of parameter
$\gamma =  \sqrt{n / (4 e^6 m)}$
has no more than $\log n$ agents.
The proof is similar to the argument used
in Lemma~\ref{lemm:NoBigIslands} for the broadcasting scenario:
\begin{lemm}
\label{lemm:NoBigNeighborhood}
Let $\gamma = \sqrt{n / (4 e^6 m)}$.
Then, the probability that there exists
an island of parameter $\gamma$
in any time instant $0 \leq t \leq 8 n \log^2 n$
with more than $\log n$ agents is at most $1/n^2$.
\end{lemm}
\begin{proof}
Since at any given time the agents are uniformly distributed in $G_n$,
the probability that a given agent is within distance $\gamma$
of another given agent at time $t_0$ is bounded by $4 \gamma^2 / n$.
Fix a time instant $t_0$ and let $\B_k(t_0)$ denote the event that
there exists an island with at least $k>\log n$ elements at time $t_0$. Then,
\[
	\prob{\B_k(t_0)}
		\leq \binom{m}{k} k^{k-2} \left(\frac{4 \gamma^2}{n}\right)^{k-1}
		\leq \left(\frac{em}{k}\right)^k k^{k-2} \left(\frac{4 \gamma^2}{n}\right)^{k-1}.
\]

Using definition of $\gamma$ and the bound $k \geq 1 + \log n$ and $m \leq n$, we have 
\[
	\prob{\B_k(t_0)} \leq \frac{em}{k^2} e^{-5(k-1)}\leq \frac{e n}{k^2} \frac{1}{n^5}
		\leq \frac{1}{n^4},
\]
for a sufficiently large $n$.
Applying the union bound on $\bigO{n \log^2 n}$ time instants concludes the proof.
\end{proof}

Next we show that, with high probability,
the frontier of the informed area cannot advance too fast.
\begin{lemm}
\label{lemm:SlowFrontier}
Let $\gamma =  \sqrt{n / (4 e^6 m)}$ and let  $t_0$
and $t_1 = t_0 + \gamma^2/(36 \log n)$ be two time steps. 
Then, with probability $1 - 1/n^2$,
\[
	\dist{\point{x}(t_1) - \point{x}(t_0)} \leq (\gamma \log n) / 2.
\]
\end{lemm}
\begin{proof}
Consider the island $I = \isla{{a_0}}{\gamma}{t_0}$
of the informed agent $a_0$ located at node $\point{x}(t_0)$ at time $t_0$.

By Lemma~\ref{lemm:props}, with probability $1-2/n^3$ an agent cannot cover
a distance of more than $\gamma / 2$ in $\gamma^2/(36 \log n)$ time steps.
Thus, with probability $1-1/n^2$, up to time $t_1$ the rumor
cannot propagate (directly or through intermediate agents) outside the set of agents $I$,
and the farthest it can get from  $\point{x}(t_0)$ by members of $I$
is bounded by $(\gamma \log n) / 2$
since the island has no more than $\log n$ agents
by Lemma~\ref{lemm:NoBigNeighborhood}.

Then, $\dist{\point{x}(t_1) - \point{x}(t_0)} \leq (\gamma \log n)/ 2$
with probability $1 - 1/n^2$.
\end{proof}

Finally, we can prove the main theorem of the subsection:
\begin{theo}
\label{theo:LBSpreadingTime2}
With high probability,
\[
	\gt = \bigOm{\frac{n}{\sqrt{m}\log^2 n}}.
\]
\end{theo}
\begin{proof}
By Lemma~\ref{lemm:Distance}, with high probability
there exists an agent $a$ whose distance at time $0$
from the source was at least $\sqrt{n}/2$.

Let $T = n / (144 e^3 \sqrt{m} \log^2 n)$ and $\gamma =  \sqrt{n / (4 e^6 m)}$.
By Lemma~\ref{lemm:SlowFrontier},
with probability $1 - 1/n$ the frontier cannot move right in $T$ steps more than
$(\gamma \log n / 2) T / (\gamma^2 / (36 \log n)) = \sqrt{n} / 4$.

By Lemma~\ref{lemm:props}, with probability $1 - 2/n^2$, agent $a$ cannot move left more than 
$2 \sqrt{T\log n} < \sqrt{n}/4$ and thus that agent is not informed at time $T$.
Hence, the gossiping time is at least $\gt > T = \bigOm{n/(\sqrt{m}\log^2 n)}$ with probability $1 - 1/n$.
\end{proof}

\noindent {\bf Acknowledgment:} Many thanks to Jeff Steif for referring us to some crucial references.

\phantomsection
\addcontentsline{toc}{chapter}{References}
\bibliographystyle{siam}
\bibliography{references}

\begin{thebibliography}{10}

\bibitem{AldousF98}
{\sc D.~Aldous and J.~Fill}, {\em Reversible Markov Chains and Random Walks on
  Graphs}, Unpublished manuscript, 1998.

\bibitem{AlonAKKLT08}
{\sc N.~Alon, C.~Avin, M.~Kouck\'y, G.~Kozma, Z.~Lotker, and M.~R. Tuttle},
  {\em Many random walks are faster than one}, in Proc. SPAA, 2008,
  pp.~119--128.

\bibitem{AlvesMP02}
{\sc O.~S.~M. Alves, F.~P. Machado, and S.~Y. Popov}, {\em The shape theorem
  for the frog model}, The Annals of Applied Probability, 12 (2002),
  pp.~533--546.

\bibitem{BroderKRU94}
{\sc A.~Z. Broder, A.~R. Karlin, P.~Raghavan, and E.~Upfal}, {\em Trading space
  for time in undirected $s-t$ connectivity}, SIAM Journal of Computing, 23
  (1994), pp.~324--334.

\bibitem{ChandraRRST97}
{\sc A.~K. Chandra, P.~Raghavan, W.~L. Ruzzo, R.~Smolensky, and P.~Tiwari},
  {\em The electrical resistance of a graph captures its commute and cover
  times}, Computational Complexity, 6 (1997), pp.~312--340.

\bibitem{ChierichettiLP10}
{\sc F.~Chierichetti, S.~Lattanzi, and A.~Panconesi}, {\em Almost tight bounds
  for rumour spreading with conductance}, in Proc. STOC, 2010, pp.~399--408.

\bibitem{ClementiMPS09}
{\sc A.~E.~F. Clementi, A.~Monti, F.~Pasquale, and R.~Silvestri}, {\em
  Information spreading in stationary markovian evolving graphs}, in Proc.
  IPDPS, 2009, pp.~1--12.

\bibitem{ClementiPS09}
{\sc A.~E.~F. Clementi, F.~Pasquale, and R.~Silvestri}, {\em {MANETS}: High
  mobility can make up for low transmission power}, in Proc. ICALP, 2009,
  pp.~387--398.

\bibitem{DimitriouNS06}
{\sc T.~Dimitriou, S.~Nikoletseas, and P.~Spirakis}, {\em The infection time of
  graphs}, Discrete Applied Mathematics, 154 (2006), pp.~2577--2589.

\bibitem{ElsasserS09}
{\sc R.~Els\"asser and T.~Sauerwald}, {\em Tight bounds for the cover time of
  multiple random walks}, in Proc. ICALP, 2009, pp.~415--426.

\bibitem{Feller68}
{\sc W.~Feller}, {\em An Introduction to Probability Theory and Its
  Applications, Vol.~I}, Wiley, 3~ed., 1968.

\bibitem{HromkovicKPRU05}
{\sc J.~Hromkovic, R.~Klasing, A.~Pelc, P.~Ruzicka, and W.~Unger}, {\em
  Dissemination of Information in Communication Networks}, Springer, Berlin,
  2005.

\bibitem{KestenS03}
{\sc H.~Kesten and V.~Sidoravicius}, {\em A shape theorem for the spread of an
  infection}.
\newblock arXiv:math/0312511v1 [math.PR], 2003.

\bibitem{Lawler91}
{\sc G.~F. Lawler}, {\em Intersections of random walks}, Birkh\"auser, Boston,
  1991.

\bibitem{MitzenmacherU05}
{\sc M.~Mitzenmacher and E.~Upfal}, {\em Probability and Computing}, Cambridge
  University Press, Cambridge, 2005.

\bibitem{Torney86}
{\sc D.~C. Torney}, {\em Variance of the range of a random walk}, Journal of
  Statistical Physics, 44 (1986), pp.~49--66.

\bibitem{WangKK08}
{\sc Y.~Wang, S.~Kapadia, and B.~Krishnamachari}, {\em Infection spread in
  wireless networks with random and adversarial node mobilities}, in Proc.
  SIGMOBILE Workshop on Mobility Models, 2008, pp.~17--24.

\bibitem{Zuckerman92}
{\sc D.~Zuckerman}, {\em A technique for lower bounding the cover time}, SIAM
  Journal of Discrete Mathematics, 5 (1992), pp.~81--87.

\end{thebibliography}

\end{document}